\documentclass[
 reprint,
 superscriptaddress,
 amsmath,amssymb,
 aps,
 pra,
 10pt
]{revtex4-1}

\usepackage{graphicx}
\usepackage{dcolumn}
\usepackage{bm}

\usepackage{graphicx}
\usepackage{xcolor}
\usepackage{color}
\makeatletter
\let\saved@includegraphics\includegraphics
\AtBeginDocument{\let\includegraphics\saved@includegraphics}
\renewenvironment*{figure}{\@float{figure}}{\end@float}
\makeatother
\usepackage{enumerate}
\usepackage[shortlabels]{enumitem}
\usepackage{amssymb}
\usepackage{amsthm}
\usepackage{amsmath}
\usepackage[colorlinks=true]{hyperref}
\hypersetup
{
colorlinks,%
citecolor=green,%
linkcolor=blue,%
}

\usepackage{braket}
\renewcommand{\v}[1]{\ensuremath{\mathbf{#1}}} 
\newcommand{\abs}[1]{\left| #1 \right|} 
\newcommand{\trace}{\mathrm{Tr}}
 
\newcommand{\dket}[1]{\vert {#1} \rangle \! \rangle} 
\newcommand{\dbra}[1]{\langle \! \langle {#1} \vert} 
\newcommand{\tM}{{\tilde M}}

\let\baraccent=\= 
\renewcommand{\=}[1]{\stackrel{#1}{=}} 

\newcommand{\thmref}[1]{\hyperref[#1]{Theorem~\ref{#1}}}
\newcommand{\lemmaref}[1]{\hyperref[#1]{Lemma~\ref{#1}}}
\newcommand{\lemmaoneref}[1]{\hyperref[#1]{Lemma~1}}
\newcommand{\figref}[1]{\hyperref[#1]{Fig.~\ref{#1}}}
\newcommand{\figaref}[1]{\hyperref[#1]{Fig.~\ref{#1}a}}
\newcommand{\figbref}[1]{\hyperref[#1]{Fig.~\ref{#1}b}}
\newcommand{\figcref}[1]{\hyperref[#1]{Fig.~\ref{#1}c}}
\renewcommand{\eqref}[1]{\hyperref[#1]{Eq.~(\ref{#1})}}
\newcommand{\eqsref}[2]{\hyperref[#1]{Eq.~(\ref{#1})-(\ref{#2})}}
\newcommand{\appref}[1]{\hyperref[#1]{Appx.~\ref{#1}}}
\newcommand{\secref}[1]{\hyperref[#1]{Sec.~\ref{#1}}}
\newcommand{\mH}{{\mathcal H}}
\newcommand{\bC}{{\mathbb C}}
\newcommand{\fraks}{{\mathfrak s}}
\renewcommand{\Re}{{\mathrm {Re}}}
\newcommand{\bR}{{\mathbb R}}

\newcommand{\sisi}[1]{{#1}}

\newtheorem{theorem}{Theorem}
\newtheorem{lemma}{Lemma}

\makeatletter
\newcommand{\setword}[2]{%
  \phantomparagraph
  #1\def\@currentlabel{\unexpanded{#1}}\label{#2}%
}
\makeatother

\begin{document}


\title{Saturating the quantum Cram\'er-Rao bound using LOCC}

\author{Sisi Zhou}
\affiliation{Departments of Applied Physics and Physics, Yale University, New Haven, Connecticut 06511, USA}
\affiliation{Yale Quantum Institute, Yale University, New Haven, Connecticut 06520, USA}
\affiliation{Pritzker School of Molecular Engineering, The University of Chicago, Illinois 60637, USA}

\author{Chang-Ling Zou}%
\affiliation{Key Laboratory of Quantum Information, CAS, University of Science and Technology of China, Hefei, Anhui 230026, China}

\author{Liang Jiang}
\affiliation{Departments of Applied Physics and Physics, Yale University, New Haven, Connecticut 06511, USA}
\affiliation{Yale Quantum Institute, Yale University, New Haven, Connecticut 06520, USA}
\affiliation{Pritzker School of Molecular Engineering, The University of Chicago, Illinois 60637, USA}

\date{\today}

\begin{abstract}
The quantum Cram\'er-Rao bound (QCRB) provides an ultimate precision limit allowed by quantum mechanics in parameter estimation. Given any quantum state dependent on a single parameter, there is always a positive-operator valued measurement (POVM) saturating the QCRB. However, the QCRB-saturating POVM cannot always be implemented efficiently, especially in multipartite systems. In this paper, we show that the POVM based on local operations and classical communication (LOCC) is QCRB-saturating for arbitrary pure states or rank-two mixed states with varying probability distributions over fixed eigenbasis. Local measurements without classical communication, however, is not QCRB-saturating in general. 
\end{abstract}

\maketitle

\section{Introduction}
\label{sec:intro}

Quantum metrology~\cite{giovannetti2006quantum,giovannetti2011advances,degen2017quantum,braun2018quantum,pezze2018quantum,pirandola2018advances}
is the study of designing high-precision quantum sensors to estimate physical
parameters in quantum systems. 
It focuses on the ultimate precision achievable in parameter estimation, allowed by
the theory of quantum mechanics. It has wide applications ranging
from frequency spectroscopy and clocks~\cite{sanders1995optimal,bollinger1996optimal,huelga1997improvement,leibfried2004toward,giovannetti2004quantum,buvzek1999optimal,valencia2004distant,de2005quantum}
to gravitational-wave detectors and interferometry~\cite{caves1981quantum,yurke19862,berry2000optimal,higgins2007entanglement}.
Lying in the center of quantum metrology is the quantum
Cram\'er-Rao bound (QCRB)~\cite{helstrom1976quantum,helstrom1968minimum,paris2009quantum,braunstein1994statistical},
which provides a lower bound of parameter estimation error:
\begin{equation}
\delta\theta\geq\frac{1}{\sqrt{NJ(\rho_{\theta})}}.
\end{equation}
Here $\theta$ is the parameter to be estimated, e.g. magnetic field
frequency, $\delta\theta$ is the standard deviation of the $\theta$-estimator,
$\rho_{\theta}$ is the density matrix describing the quantum sensor
as a function of $\theta$, and $N$ is the number of repeated experiments.
$J(\rho_{\theta})$ is the so-called quantum Fisher information (QFI)~\cite{helstrom1976quantum,helstrom1968minimum,paris2009quantum,braunstein1994statistical}
quantifying the sensitivity of a quantum sensor. 

QFI can be viewed as the maximum Fisher information (FI) among all
possible POVMs, where FI is the classical version of QFI as a measure
of sensitivity~\cite{kobayashi2011probability,casella2002statistical,lehmann2006theory}.
It is a function of the probability distribution of measurement results.
In a quantum system, the probability distribution is provided by $P_{x}(\theta)=\trace(\rho_{\theta}E_{x})$
with measurement operators $\{E_{x}\}$. 
To saturate the QCRB, one first performs the optimal POVM maximizing
the FI~\cite{braunstein1994statistical}, and then chooses suitable
classical estimators, e.g. the maximum likelihood estimator which asymptotically ($N\gg1$) saturates the QCRB~\cite{lehmann2006theory,casella2002statistical,brody1996geometry,fujiwara2006strong}.
\sisi{The optimal POVM usually depends on the value of the parameter, which is unknown practically. In order to solve this issue, one could use the two-step method by first using $\sqrt{N}$ states to obtain a rough estimation $\tilde\theta \approx \theta$ and then performing the optimal measurement based on $\tilde \theta$ on the remaining $N-\sqrt{N}$ states~\cite{barndorff2000fisher,hayashi2011comparison,yang2019attaining}. The procedure introduces a negligible amount of error asymptotically.}
If the problem was approached using a Bayesian approach, under quite general conditions, the QFI remains a reliable figure of merit~\cite{lehmann2006theory,gill2008conciliation,pezze2014quantum,jarzyna2015true,gorecki2019pi}.
It is known that rank-one projection onto the eigenstates
of the symmetric logarithmic derivative operator (SLD) usually 
saturates the QCRB~\cite{braunstein1994statistical}. However, in general, the eigenstates of SLD could be highly-entangled states over subsystems, and the optimal measurement requires global measurements (GM) (\figaref{fig:measure}) that might be challenging to implement experimentally~\cite{friis2017flexible}. 

\begin{figure}[htbp]
\includegraphics[width=6.5cm]{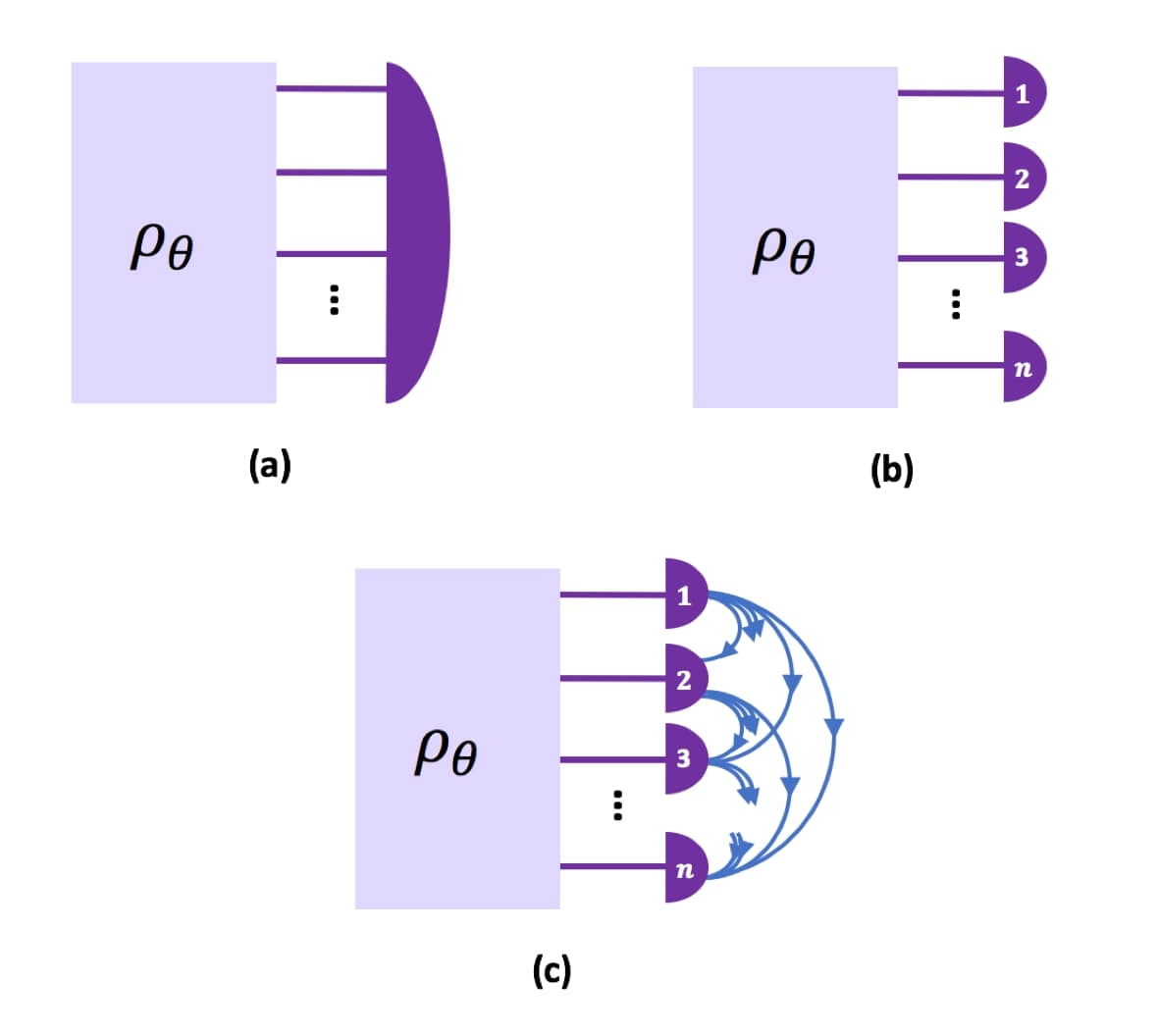}
\caption{\label{fig:measure} Schematics of measurement protocols in quantum metrology. Here, ${\rm LM} \subseteq {\rm LOCC} \subseteq {\rm GM}$. (a) Global measurements (GM). (b) Local measurements (LM). (c) Local operations and classical communication (LOCC). Blue lines represent classical data flows. The state preparation and the probing processes are not shown because they can be as general as possible. 
} 
\end{figure}

Local measurements (LM) (\figbref{fig:measure}), performed separately on each subsystem, were shown to saturate the QCRB in many cases~\cite{giovannetti2006quantum,boixo2008quantum,roy2008exponentially,rams2018limits}.
For example, it is proven in Ref.~\cite{giovannetti2006quantum}
that for GHZ-type states evolving under local Hamiltonians with identical terms,
LM can saturate the QCRB.  
However, by counting the number of degrees of freedom in LM and the QCRB-saturating condition, one can show that LM, in general, is not sufficient to saturate the QCRB in multipartite systems (see \appref{app:no-go-LM} for proof).
Compared to LM, local operations and classical communication (LOCC) (\figcref{fig:measure}) is a larger class of measurements which allows classical communication of measurement results so that the measurement basis performed on one subsystem could
be determined by the measurement results from others~\cite{chitambar2014everything,cirac1999distributed,raussendorf2001one,divincenzo2002quantum,verstraete2003quantum}, which has been demonstrated in many experimental platforms compatible with local
measurement and adaptive control~\cite{Kalb17,Chou18,bernien2017probing,zhang2017observation}. It is a restricted class of quantum operations~\cite{bagan2004collective,bagan2006separable,calsamiglia2010local} that cannot generate entanglement between subsystems. For example, it cannot fully distinguish the four Bell states~\cite{ghosh2001distinguishability}. Nevertheless, LOCC can distinguish any two orthogonal quantum states~\cite{walgate2000local} and, in particular, tell the quantum state itself from the state it evolves into, making it a potential candidate to saturate the QCRB. {
The power of LOCC protocols in achieving optimal performance has also been demonstrated in other contexts~\cite{acin2000optimal,bagan2005purity,duan2007distinguishing,yu2011any,hayashi2006bounds}. }

In this paper, we consider only quantum states $\rho_\theta$ in finite-dimensional Hilbert spaces. We prove that LOCC is QCRB-saturating for two types of quantum states: (i) arbitrary pure states $\rho_{\theta}=\ket{\psi_{\theta}}\bra{\psi_{\theta}}$ and (ii) rank-two mixed states $\rho_{\theta}=p_{\theta}\ket{\psi_{0}}\bra{\psi_{0}}+(1-p_{\theta})\ket{\psi_{1}}\bra{\psi_{1}}$ ($0 < p_\theta < 1$), where $\ket{\psi_{0,1}}$ are fixed basis independent of $\theta$. In the following, we first review the necessary and sufficient condition for QCRB-saturating
measurements in finite-dimensional Hilbert spaces. Then we prove the existence of QCRB-saturating LOCC for states of type (i) and type (ii). Finally, we show that LM is not QCRB-saturating in general. For bipartite pure states, we found an interesting example where there is no QCRB-saturating projective LM but there is a QCRB-saturating LM.

\section{QCRB-saturating POVM}
\label{sec:ns}

To quantify the distinguishability of two neighboring probability distributions, the FI is defined by
\begin{equation}
F(\{P_{  x }(\theta)\}) = \sum_{{  x }} \frac{1}{P_{  x }(\theta)}\Big(\frac{\partial P_{  x }(\theta)}{\partial \theta}\Big)^2,
\end{equation}
where ${  x }$ is the label of measurement results, $P_{  x }(\theta)$ is the probability of obtaining ${  x }$ when the parameter is equal to $\theta$, satisfying $P_{  x }(\theta) \geq 0$ and $\sum_{  x } P_{  x }(\theta) = 1$. For a quantum state $\rho_\theta$, $P_x(\theta)=\trace(\rho_\theta E_x)$ for a POVM described by a set of non-negative operators $\{E_{  x }\}$ satisfying $\sum_{  x } E_{  x } = I$, and the FI
\begin{equation}
\label{eq:inequality}
F(\{P_{  x }(\theta)\}) 
\leq \trace(\rho_\theta L_\theta^2) = J(\rho_\theta).
\end{equation}
Here, $L_\theta$ is the SLD, a Hermitian matrix defined by $\partial_\theta\rho_\theta = \frac{1}{2}(L_\theta \rho_\theta + \rho_\theta L_\theta)$. The FI is equal to the QFI $J(\rho_\theta)$ if and only if,
\begin{equation}
\label{eq:ns-orig}
E_{  x }^{1/2} \rho_\theta^{1/2} = \lambda_{  x } E_{  x }^{1/2} L_\theta \rho_\theta^{1/2},\quad\forall x,
\end{equation}
for some real $\lambda_x$ and for all $x$ such that $\trace(E_x\rho_\theta) = 0$, $\trace(E_x L_\theta \rho_\theta L_\theta) = 0$. We call any POVM $\{E_{x}\}$ satisfying \eqref{eq:ns-orig} \emph{QCRB-saturating}.
Further simplifications of \eqref{eq:ns-orig} leads to~(see \appref{app:ns})

\begin{theorem}[\cite{braunstein1994statistical}]
\label{thm:ns}
$\{E_x\}$ is QCRB-saturating if and only if 
\begin{equation}
\label{eq:ns}
E_x^{1/2} M_{ij} E_x^{1/2} = 0,\quad\forall i,j,x,
\end{equation}
and 
\begin{equation}
\label{eq:reg}
\forall x\text{~s.t.~}\trace(E_x\rho_\theta) = 0,\quad E_x^{1/2} L_\theta \ket{\psi_{\theta,i}} = 0,\;\forall i.
\end{equation}
Here we use the diagonalization of the density matrix $\rho_\theta = \sum_{k} p_{\theta,k} \ket{\psi_{\theta,k}}\bra{\psi_{\theta,k}}$ ($p_{\theta,k} > 0$) and 
\begin{equation}
M_{ij} = \ket{\psi_{\theta,i}}\bra{\psi_{\theta,j}} L_\theta - L_\theta \ket{\psi_{\theta,i}}\bra{\psi_{\theta,j}}. 
\end{equation}
\end{theorem}

The condition \eqref{eq:reg}, though not explicitly spelled out in \cite{braunstein1994statistical}, is necessary in order to deal with measurements satisfying $\trace(\rho_x E_x) = 0$ \cite{zhu2018universally}. 
From \thmref{thm:ns}, it is clear that rank-one projection onto the eigenstates of $L_\theta$ satisfies \eqref{eq:ns}. As an example, we consider sensing with $n$-partite GHZ states
\begin{equation}
\label{eq:GHZ}
\ket{\psi_\theta} = \frac{1}{\sqrt{2}} \left(\ket{0^{\otimes n}} + e^{in\theta}\ket{1^{\otimes n}}\right),
\end{equation}
which can be viewed as the evolution of $\ket{\psi_{\rm in}} = \frac{1}{\sqrt{2}} \left(\ket{0^{\otimes n}} + \ket{1^{\otimes n}}\right)$ under the Hamiltonian $H = \frac{\theta}{2}\sum_{i=1}^n  \sigma^{z,i}$ after unit time, where $\sigma^{z,i}$ is the Pauli-Z matrix acting on the $i$-th qubit. The corresponding SLD is
\begin{equation}
L_\theta = n(ie^{in\theta}\ket{1^{\otimes n}}\bra{0^{\otimes n}}-ie^{-in\theta}\ket{0^{\otimes n}}\bra{1^{\otimes n}}),
\end{equation}
whose eigenstates $\frac{1}{\sqrt{2}}\left(\ket{0^{\otimes n}} \pm ie^{in\theta} \ket{1^{\otimes n}}\right)$ also satisfies \eqref{eq:reg} and therefore induce a QCRB-saturating measurement. 
Saturating the QCRB using projective measurements onto these maximally entangled states requires coupling gates between subsystems and might be challenging for practical experimental implementations. Alternatively, it is well known that projection onto $\ket{\pm} = \frac{1}{\sqrt 2}(\ket{0} \pm \ket{1})$ of individual qubits is also QCRB-saturating~\cite{giovannetti2006quantum}. However, the systematic approach to identify experimental-friendly QCRB-saturating POVM have never been discussed before.

\section{LOCC protocol} 
\label{sec:LOCC}

{
For arbitrary quantum states, LOCC is not suffcient to saturate the QCRB.  
Consider the following two-qubit quantum state
\begin{equation}
\rho_\theta = \theta \rho_1 + (1-\theta) \rho_2,
\end{equation}
where 
\begin{gather}
\rho_1 = \frac{2}{3}\ket{\beta_1}\bra{\beta_1}
+ \frac{1}{3}\ket{\beta_2}\bra{\beta_2} ,\\
\rho_2 =  \frac{1}{3}\ket{\beta_1}\bra{\beta_1} + \frac{2}{3}\ket{\beta_3}\bra{\beta_3},
\end{gather}
$\{\ket{\beta_{i}}\}_{i=1,2,3}$ are three of the Bell states (we don't care about the order of the labels). 
The SLD operator is 
\begin{equation}
L_\theta = 
\frac{1}{1+\theta}\ket{\beta_1}\bra{\beta_1}
+ \frac{1}{\theta}\ket{\beta_2}\bra{\beta_2} 
+ \frac{-1}{(1-\theta)}\ket{\beta_3}\bra{\beta_3},
\end{equation}
whose coefficients of $\ket{\beta_i}\bra{\beta_i}$ are all different. Therefore $\{M_{ij}\}$ in \eqref{eq:ns} contains terms proportional to $\ket{\beta_i}\bra{\beta_j}$ for all $i\neq j$. If there is an LOCC such that \eqref{eq:ns} is satisfied, then
\begin{equation}
E_x^{1/2} \ket{\beta_i}\bra{\beta_j} E_x^{1/2} = 0,\quad\forall i,j,x,
\end{equation} 
contradicting the fact that any three Bell states cannot be distinguished from each other using LOCC~\cite{ghosh2001distinguishability}. Therefore, LOCC cannot saturate the QCRB for $\rho_\theta$. 
}

Now we consider LOCC as potential candidates to saturate the QCRB for the following two types of quantum states: (i) arbitrary pure states $\rho_\theta = \ket{\psi_\theta}\bra{\psi_\theta}$, which is one of the most commonly used states in quantum metrology \cite{giovannetti2006quantum,giovannetti2011advances}; and (ii) rank-two mixed states $\rho_\theta = p_{\theta}\ket{\psi_0}\bra{\psi_0} + (1-p_{\theta}) \ket{\psi_1}\bra{\psi_1}$, where $\ket{\psi_{0,1}}$ are independent of $\theta$, which might find applications in quantum thermometry~\cite{de2016local,sone2018quantifying,sone2018nonclassical}. These states only have one distinct $M$ in \eqref{eq:ns}. For type (i) states, \eqref{eq:ns} and \eqref{eq:reg} becomes
\begin{equation}
\label{eq:main-ns}
E_x^{1/2} M E_x^{1/2} = 0,\quad\forall x,
\end{equation}
and 
\begin{equation}
\label{eq:main-reg}
\forall x \text{~s.t.~} E_x^{1/2}\ket{\psi_\theta} = 0,\quad E_x^{1/2}\ket{\psi_\theta^\perp} = 0, 
\end{equation}
where $\ket{\psi_\theta^\perp} := (1-\ket{\psi_\theta}\bra{\psi_\theta})\ket{\partial_\theta \psi}$ and 
\begin{equation}
M
= \ket{\psi_\theta}\bra{\psi_\theta^\perp} - \ket{\psi_\theta^\perp}\bra{\psi_\theta}
\end{equation} is a traceless anti-Hermitian matrix. For type (ii) states, we have $M = \ket{\psi_0}\bra{\psi_1}$ and \eqref{eq:reg} is always satisfied. 

In particular, for rank-one projective measurements where $E_x = \ket{E_x}\bra{E_x}$, \eqref{eq:ns} becomes $\braket{E_x|M|E_x} = 0$, $\forall x$. Let us define 
\begin{equation}
\tilde{M} = 
\begin{cases}
M + i\left(\ket{\psi_\theta}\bra{\psi_\theta} - \frac{I}{D}\right)&  \text{for type (i)},\\
M & \text{for type (ii)},\\
\end{cases}
\end{equation}
where $D$ is the dimension of the entire Hilbert space. For type (i) states, $\braket{E_x|\tilde M|E_x} = 0$ implies $\braket{E_x|M|E_x} = 0$ and $\abs{\braket{E_x|\psi_\theta}}^2 = 1/D$ for all $x$. \eqref{eq:main-ns} and \eqref{eq:main-reg} are satisfied. Thus $\braket{E_x|\tilde M|E_x} = 0$, $\forall x$ is a sufficient condition for $\{\ket{E_x}\}$ to be QCRB-saturating. By constructing a LOCC measurement basis satisfying $\braket{E_x|\tilde M|E_x} = 0$, $\forall x$, we prove the following theorem:

\begin{theorem}
\label{thm:locc}
For any multipartite state belonging to type (i) and (ii), there exists a QCRB-saturating LOCC measurement protocol. 
\end{theorem}

In fact, the QCRB is saturable for arbitrary $n$ when one-way classical communication is allowed, where the measurement result of subsystem ${\mathfrak s}_k$ is classically communicated to ${\mathfrak s}_{k+1},\ldots,{\mathfrak s}_{n}$ to assist the choice of their measurement basis. The corresponding POVM (\figcref{fig:measure}) is
\begin{equation}
\label{eq:LOCC}
E_{x_1,\ldots,x_n} = E_{x_1}^{{\mathfrak s}_1} \otimes E_{x_1,x_2}^{{\mathfrak s}_2} \otimes \cdots \otimes E_{x_1,\ldots,x_n}^{{\mathfrak s}_n},
\end{equation}
where $E^{{\mathfrak s}_k}_{x_1,\cdots,x_k}$ are non-negative operators in subsystem ${\mathfrak s}_k$ satisfying $\sum_{x_k} E_{x_1,\cdots,x_k}^{{\mathfrak s}_k} = I^{{\mathfrak s}_k}$. 

The procedure to construct a QCRB-saturating rank-one projective LOCC, where $E_{x_1,\ldots,x_k}^{{\mathfrak s}_k} = \ket{E_{x_1,\ldots,x_k}^{{\mathfrak s}_k}}\bra{E_{x_1,\ldots,x_k}^{{\mathfrak s}_k}}$, with the structure in \eqref{eq:LOCC} can be summarized as follows: 
\begin{enumerate}[(1),wide,labelwidth=!,labelindent=0pt]
\item Calculate $\tM^{{\mathfrak s}_1} = \trace_{{\mathfrak s}_2\cdots {\mathfrak s}_n}(\tM)$ by tracing out subsystems $\{{\mathfrak s}_2,\ldots,{\mathfrak s}_n\}$ in matrix $\tM$; 
\item Find an orthonormal basis $\ket{E^{{\mathfrak s}_1}_{x_1}}$ in ${\mathfrak s}_1$ such that $\braket{E^{{\mathfrak s}_1}_{x_1}|M^{{\mathfrak s}_1}|E^{{\mathfrak s}_1}_{x_1}} = 0$; 
\item Calculate $\tM^{{\mathfrak s}_2}_{x_1} = \trace_{{\mathfrak s}_3\cdots {\mathfrak s}_n}(\braket{E^{{\mathfrak s}_1}_{x_1}|\tM|E^{{\mathfrak s}_1}_{x_1}})$; 
\item Find an orthonormal basis $\ket{E^{{\mathfrak s}_2}_{x_1,x_2}}$ in ${\mathfrak s}_2$ such that $\braket{E^{{\mathfrak s}_2}_{x_1,x_2}|\tM_{x_1}^{{\mathfrak s}_2}|E^{{\mathfrak s}_2}_{x_1,x_2}} = 0$; 
\item Repeat steps (3)-(4) for subsystems ${\mathfrak s}_3$,\ldots,${\mathfrak s}_n$. 
\end{enumerate}

In steps (2) and (4), we use the following lemma:

\begin{lemma}
\label{thm:zero-diag}
Given any traceless matrix $\tilde{M} \in \bC^{d\times d}$, there exists an complete orthonormal basis $\{\ket{u_i}\}_{i=1}^d$ in $\bC^{d}$ such that $\bra{u_i}\tilde{M}\ket{u_i} = 0$ for all $i$. 
\end{lemma}

A constructive proof can be found in \appref{app:LOCC}. Our construction is mathematically reminiscent of the one provided in Ref.~\cite{walgate2000local} where LOCC is used to distinguish two multipartite orthogonal quantum states, but our construction does not require extending the dimension of each subsystem to be a power of two. In fact, parameter estimation is closely related to state discrimination. Projective measurements $\{\ket{E_{  x }}\bra{E_{  x }}\}$ distinguishs two orthogonal quantum states $\ket{\psi_{0,1}}$ as long as $\braket{E_{  x }|\psi_0}\braket{\psi_1|E_{  x }} = 0$~\appref{app:dist}. 
It is then clear that a measurement distinguishing an orthonormal basis $\{\ket{\psi_{k}}\}$ is also QCRB-saturating when estimating $\theta$ in the probability coefficients for any mixed quantum states $\sum_{k} p_{\theta,k} \ket{\psi_{k}}\bra{\psi_{k}}$ ($p_{\theta,k} > 0$). 

In \appref{app:example}, we provide an example of a four-qubit system with a nearest neighbour interaction Hamiltonian, where the parameter to estimated is the strength of the Hamiltonian. We use the algorithm described above to calculate the LOCC measurement basis and plot them in the Bloch spheres.

{
\section{Local measurements}
\label{sec:LM}

LM in a $n$-partite system $\{{\mathfrak s}_1,\ldots,{\mathfrak s}_n\}$ (\figbref{fig:measure}) has the following structure
\begin{equation}
\label{eq:LM}
E_{x_1,\ldots,x_n} = E_{x_1}^{{\mathfrak s}_1} \otimes E_{x_2}^{{\mathfrak s}_2} \otimes \cdots \otimes E_{x_n}^{{\mathfrak s}_n},
\end{equation}
where $x_k$ is the $k$-th measurement result and $\{E_{x_k}^{{\mathfrak s}_k}\}$ is a POVM in subsystem ${\mathfrak s}_k$. 
One may wonder whether LM would be sufficient to saturate the QCRB, as for GHZ states. It is not possible in general for sufficiently large $n$, because the number of the degrees of freedom in LM grows linearly as the number of qubits increases but that in the quantum states grows exponentially (see the detailed proof in \appref{app:no-go-LM}). 

For bipartite pure states, however, the argument above does not hold and the problem should be treated carefully. Consider $\ket{\psi_\theta} \in \mH_1 \otimes \mH_2$ where $\dim \mH_1 = d_1$ and $\dim \mH_2 = d_2$. Let 
\begin{equation}
\ket{\psi_\theta} = \sum_{ij} A_{ij}\ket{i}\ket{j},\quad \ket{\psi_\theta^\perp} = \sum_{ij} B_{ij}\ket{i}\ket{j}. 
\end{equation}
The orthogonality condition $\braket{\psi_\theta|\psi_\theta^\perp} = 0$ implies $\trace(A^\dagger B) = 0$. 
Then we have the following lemma: 
\begin{lemma}
\label{thm:bi}
There exists a QCRB-saturating LM for $\ket{\psi_\theta}$ if and only if there are isometries $U$ and $V$ satisfying $U U^\dagger = I \in \bC^{d_1 \times d_1}$ and $V V^\dagger = I \in \bC^{d_2 \times d_2}$ such that $C = U^\dagger AV$, $D = U^\dagger BV$ satisfying ($\overline{\,\cdot\,}$ means complex conjugate)
\begin{equation}
\label{eq:bipartite-1}
C_{ij}\overline{D}_{ij} = \overline{C}_{ij} D_{ij},\quad\forall i,j, 
\end{equation}
and 
\begin{equation}
\label{eq:bipartite-2}
\forall i,j,\quad \text{if~~}C_{ij} = 0,\quad D_{ij} = 0. 
\end{equation}
\end{lemma}
\begin{proof}
On one hand, given $U$ and $V$ satisfying \eqref{eq:bipartite-1} and \eqref{eq:bipartite-2}, let $U = \sum_{i}\ket{E^{{\mathfrak s}_1}_i}\bra{i}$, $V = \sum_j \ket{\overline{E^{{\mathfrak s}_2}_j}}\bra{j}$. Then $\{E^{{\mathfrak s}_1}_i \otimes E^{{\mathfrak s}_2}_j = \ket{E^{{\mathfrak s}_1}_i}\bra{E^{{\mathfrak s}_1}_i} \otimes \ket{E^{{\mathfrak s}_2}_j}\bra{E^{{\mathfrak s}_2}_j} ,\forall i,j\}$ is a LM and the QCRB-saturating conditions \eqref{eq:ns} and \eqref{eq:reg} become
\begin{multline}
(\bra{E^{{\mathfrak s}_1}_i}\otimes\bra{E^{{\mathfrak s}_2}_j})\ket{\psi_\theta}\bra{\psi_\theta^\perp}(\ket{E^{{\mathfrak s}_1}_i}\otimes\ket{E^{{\mathfrak s}_2}_j}) = \\
(\bra{E^{{\mathfrak s}_1}_i}\otimes\bra{E^{{\mathfrak s}_2}_j})\ket{\psi_\theta^\perp}\bra{\psi_\theta}(\ket{E^{{\mathfrak s}_1}_i}\otimes\ket{E^{{\mathfrak s}_2}_j}),\forall i,j,
\end{multline}
and 
\begin{equation}
(\bra{E^{{\mathfrak s}_1}_i}\otimes\bra{E^{{\mathfrak s}_2}_j})\ket{\psi_\theta^\perp} = 0, \text{~if~} (\bra{E^{{\mathfrak s}_1}_i}\otimes\bra{E^{{\mathfrak s}_2}_j})\ket{\psi_\theta} = 0,
\end{equation}
which are equivalent to \eqref{eq:bipartite-1} and \eqref{eq:bipartite-2}. Here $\{\ket{E^{{\mathfrak s}_1}_i}\}$ and $\{\ket{E^{{\mathfrak s}_2}_j}\}$ are not unit vectors in general. When they are, $U$ and $V$ are unitary operators and give rise to a QCRB-saturating rank-one projective LM. 

On the other hand, given a QCRB-saturating LM $\{E_i^{{\mathfrak s}_1} \otimes E_i^{{\mathfrak s}_2}\}$, let $E_i^{{\mathfrak s}_{1}} = \sum_{k} \ket{E_{i,k}^{{\mathfrak s}_{1}}}\bra{E_{i,k}^{{\mathfrak s}_{1}}}$ and $E_j^{{\mathfrak s}_2} = \sum_{k} \ket{E_{j,k}^{{\mathfrak s}_2}}\bra{E_{j,k}^{{\mathfrak s}_2}}$ where $\braket{E_{i,k}^{{\mathfrak s}_{1,2}}|E_{i,k'}^{{\mathfrak s}_{1,2}}} = 0$ if $k\neq k'$ and positive if $k=k'$. Then $U = \sum_{i,k}\ket{E^{{\mathfrak s}_1}_{i,k}}\bra{i,k}$ and $V = \sum_{j,k} \ket{\overline{E^{{\mathfrak s}_2}_{j,k}}}\bra{j,k}$ satisfy \eqref{eq:bipartite-1} and \eqref{eq:bipartite-2}. 
\end{proof}

Using \lemmaref{thm:bi}, we make the following observations on the QCRB-saturating LM for bipartite pure states:
\begin{enumerate}[(a),wide,labelwidth=!,labelindent=0pt]

\item Parameter estimation is not equivalent to orthogonal state discrimination --- there exists $\ket{\psi_\theta}$ and $\ket{\psi_\theta^\perp}$ such that they cannot be distinguished using LM, but there exists a QCRB-saturating LM for them. We show in \appref{app:lm-1} that
\begin{equation}
\label{eq:lm-1}
\ket{\psi_\theta} = \frac{1}{\sqrt{2}}(\ket{00}+\ket{1+}),\quad \ket{\psi_\theta^\perp} = \frac{1}{\sqrt{2}}(\ket{01}+\ket{1-})
\end{equation} 
is the desired example. 
Note that this also implies LM is not always QCRB-saturating for type (ii) bipartite states. 


\item LM is not QCRB-saturating for all bipartite pure states --- \eqref{eq:ns} and \eqref{eq:reg} cannot always be satisfied simultaneously. Consider a two-qubit system where $A = I/2$ and $B = e^{\frac{i\pi}{4}}\sigma_x/2$. Suppose \eqref{eq:bipartite-1} and \eqref{eq:bipartite-2} are both satisfied for some $U$ and $V$. Then $U^\dagger( A B^\dagger - B A^\dagger )U$ and $V^\dagger( A^\dagger B - B^\dagger A )V$ are zero-diagonal (i.e. all diagonal elements are zero), which implies $\Re[\overline{U}_{1i}U_{2i}] = \Re[\overline{V}_{1j}V_{2j}] = 0$, $\forall i,j$. Without loss of generality, assume $U_{1i},V_{1j}\in\bR$, $U_{2i},V_{2j}\in i\bR$. Then $C_{ij}\in\bR$, $D_{ij}\in e^{i\frac{3\pi}{4}}\bR$. Therefore $C_{ij}$ and $D_{ij}$ cannot be simultaneously non-zero. According to \eqref{eq:bipartite-2}, we must have $D_{ij} = 0$ for all $i,j$ which is not possible. 

\item \eqref{eq:ns} itself can be satisfied for all $A,B \in \bC^{2\times d}$, $d \geq 2$. According to \lemmaref{thm:zero-diag}, there exists a unitary matrix $U = \sum_{i=1}^2 \ket{E^{{\mathfrak s}_1}_i}\bra{i}$ such that $U^\dagger (AB^\dagger - BA^\dagger) U$ is zero-diagonal. Furthermore, according to \lemmaref{thm:zero-diag}, there exists a unitary matrix $V = \sum_{j=1}^d \overline{\ket{E^{{\mathfrak s}_1}_j}}\bra{j}$ such that both $B^\dagger\ket{E^{{\mathfrak s}_1}_1}\bra{E^{{\mathfrak s}_1}_1}A - A^\dagger\ket{E^{{\mathfrak s}_1}_1}\bra{E^{{\mathfrak s}_1}_1}B$ and $B^\dagger\ket{E^{{\mathfrak s}_1}_2}\bra{E^{{\mathfrak s}_1}_2}A - A^\dagger\ket{E^{{\mathfrak s}_1}_2}\bra{E^{{\mathfrak s}_1}_2}B$ are zero-diagonal. Now we have the desired $U$ and $V$. In practice, one could first find $U$ and $V$ using this procedure and check whether \eqref{eq:bipartite-2} is also satisfied. If so, we have a QCRB-saturating LM. The $U$ and $V$ constructed here are both unitary and the corresponding LM is projective.

\item Projective LM is distinct from general LM --- When $A,B \in \bC^{3\times 3}$, an example exists where there is no QCRB-saturating projective LM, but there is a QCRB-saturating LM. So far, we have shown that projective LOCC is sufficient to saturate the QCRB for type (i) and (ii) states and projective LM is sufficient to satisfy \eqref{eq:ns} for bipartite pure states when $A, B\in \bC^{2\times d}$. However, when $A,B \in \bC^{3 \times 3}$, projective LM is not as powerful as general LM. We show in \appref{app:lm-2} that
\begin{equation}
\label{eq:3-3}
A = \begin{pmatrix}
{\sqrt{2}}/{2} & 0 & 0 \\
0 & {1}/{2} & 0\\
0 & 0 & 1/2\\
\end{pmatrix},\;
B = \begin{pmatrix}
{\sqrt{2}i}/{2} & 0 & 0 \\
0 & {-i}/{2} & 0\\
0 & 0 & {-i}/{2}\\
\end{pmatrix} 
\end{equation} 
is the desired example. 

\end{enumerate}

To sum up, we have found a bipartite pure state where there is no QCRB-saturating LM. We also show that the QCRB saturability problem for pure bipartite states distinguishs projective LM from general LM. However, it is not clear whether our examples could be generalized. It is an interesting open question to classify bipartitue pure states by the existence of the QCRB-saturating LM. 



\section{Conclusion and Outlook}

We have investigated the QCRB-saturating measurement to maximize the
sensitivity of quantum sensors. For arbitrary pure states or rank-two mixed states with fixed eigenbasis, we
have developed the QCRB-saturating LOCC protocol, feasible with many physical platforms by local
measurement and adaptive control~\cite{Kalb17,Chou18,bernien2017probing,zhang2017observation}. Our LOCC protocol may have applications in extensive parameter estimation
and calibration scenarios, including criticality-based quantum metrology~\cite{venuti2007quantum,rams2018limits},
quantum thermometry~\cite{de2016local,sone2018quantifying,sone2018nonclassical} and various
other cases in many body physics~\cite{tamascelli2016characterization,hauke2016measuring,zhang2018characterization}.

Our LOCC sensing protocol crucially relies on the fact that two orthogonal states can be distinguished using LOCC, so that it can be QCRB-saturating for pure states or rank-two mixed states with fixed eigenbasis. In practice, the quantum states could suffer various decoherences and our protocol might not be able to saturate the QCRB for general mixed states or for multi-parameter sensing \cite{yang2018optimal,matsumoto2002new,ragy2016compatibility,baumgratz2016quantum,yuan2016sequential,proctor2018multiparameter,ge2018distributed,qian2019heisenberg}. To tackle the decoherence, we may apply dynamical decoupling to suppress time-correlated noises~\cite{viola1999dynamical,uhrig2007keeping,biercuk2009optimized}, or introduce quantum error correction to restore unitary evolution in logical subspace even in the presence of Markovian noises~\cite{kessler2014quantum,dur2014improved,arrad2014increasing,zhou2018achieving,demkowicz2017adaptive,reiter2017dissipative}.  Therefore, it will be intriguing to further investigate LOCC sensing protocol combined with quantum error correction.

\section{Acknowledgements}
We thank the anonymous reviewer for pointing out the loophole in Theorem 1 which is now fixed by adding \eqref{eq:reg}. We thank Steven Flammia, Arpit Dua, Wen-long Ma, Shengjun Wu, Zi-wen Liu and Yau Wing Li for helpful discussions. 
We acknowledge support from the ARL-CDQI (W911NF-15-2-0067), ARO (W911NF-14-1-0011, W911NF-14-1-0563), ARO
MURI (W911NF-16-1-0349 ), AFOSR MURI (FA9550-14-1-0052, FA9550-15-1-0015), NSF (EFMA-1640959), Alfred P. Sloan Foundation (BR2013-049), and Packard Foundation (2013-39273).

\bibliographystyle{apsrev4-1}
\bibliography{local-refs}

\onecolumngrid
\newpage

\appendix



\section{\label{app:ns}The necessary and sufficient condition~for QCRB-saturating POVM}

In this appendix, we prove \thmref{thm:ns} in the main text. The classical Fisher information $F(\rho_\theta)$ satisfies
\begin{equation}
\begin{split}
F(\rho_\theta) &= \sum_{  x:\trace(E_x\rho_\theta) \neq 0 } \frac{\big(\trace(E_{  x } \partial_\theta\rho_\theta)\big)^2}{\trace(E_{  x }\rho_\theta)}
= \sum_{  x:\trace(E_x\rho_\theta) \neq 0  } \frac{\big({\rm Re}[\trace(E_{  x } L_\theta \rho_\theta)]\big)^2}{\trace(E_{  x }\rho_\theta)}\\
&\leq \sum_{  x:\trace(E_x\rho_\theta) \neq 0  } \frac{\big(\abs{\trace(E_{  x } L_\theta \rho_\theta)}\big)^2}{\trace(E_{  x }\rho_\theta)} \leq \sum_{  x :\trace(E_x\rho_\theta) \neq 0 } \trace(E_{  x } L_\theta \rho_\theta  L_\theta  )
\\&\leq \trace(L_\theta^2 \rho_\theta) \equiv J(\rho_\theta),
\end{split}
\end{equation}
where the first equality holds true when
\begin{equation}
\label{eq:ns-1}
{\rm Im}[\trace(E_{  x } L_\theta \rho_\theta)] = 0,\;\text{for all }{  x },
\end{equation}
the second equality holds true when 
\begin{equation}
\label{eq:ns-2}
E_{  x }^{1/2} \rho_\theta^{1/2} = \lambda_{  x } E_{  x }^{1/2} L_\theta \rho_\theta^{1/2},\;\lambda_{  x } \in \mathbb{C},\;\text{for all }{  x },
\end{equation}
based on the use of the Cauchy-Schwarz inequality,
and the third equality holds true when 
\begin{equation}
\forall x \, \text{~s.t.~} \trace(E_x \rho_\theta) = 0,\,\,\quad \trace(E_x L_\theta \rho_\theta L_\theta) = 0 \Leftrightarrow E_x^{1/2} L_\theta \ket{\psi_{\theta,i}} = 0,~\forall i. 
\end{equation}

For pure states,
\begin{gather}
L_\theta = 2(\ket{\partial_\theta\psi_\theta}\bra{\psi_\theta}+\ket{\psi_\theta}\bra{\partial_\theta\psi_\theta}),
\\
J(\rho_\theta) = 4\abs{\braket{\partial_\theta\psi_\theta|\psi_\theta}}^2.
\end{gather}
and in general when $\rho_\theta = \sum_{k} p_{\theta,k} \ket{\psi_{\theta,k}}\bra{\psi_{\theta,k}}$, 
\begin{gather}
L_\theta = \sum_{\substack{j,k\\ p_{\theta,j}+p_{\theta,k} \neq 0}}\frac{2}{p_{\theta,j}+p_{\theta,k}} \bra{\psi_{\theta,j}} \partial_\theta \rho_\theta  \ket{\psi_{\theta,k}} \ket{\psi_{\theta,j}}\bra{\psi_{\theta,k}},\\
J(\rho_\theta) = \sum_{\substack{j,k\\ p_{\theta,j}+p_{\theta,k} \neq 0}}\frac{2}{p_{\theta,j}+p_{\theta,k}} \abs{\bra{\psi_{\theta,j}} \partial_\theta \rho_\theta  \ket{\psi_{\theta,k}}}^2 .
\end{gather}
Combining \eqref{eq:ns-1} and \eqref{eq:ns-2}, we get
\begin{equation}
\label{eq:ns-3}
E_{  x }^{1/2} \rho_\theta^{1/2} = \lambda_{  x } E_{  x }^{1/2} L_\theta \rho_\theta^{1/2},\;\lambda_{  x } \in \mathbb{R},\;\text{for all }{  x }.
\end{equation}
Therefore \eqref{eq:ns-3} is a necessary and sufficient condition for a POVM $\{E_{  x }\}$ to be QCRB-saturating. 

To eliminate $\lambda_{  x }$ in \eqref{eq:ns-3}, one may first rewrite it via vectorization:
\begin{equation}
\label{eq:ns-5}
(E^{1/2}_{  x } \otimes I) \dket{\rho_\theta^{1/2}} = \lambda_{  x } (E^{1/2}_{  x } \otimes I) \dket{ L_\theta \rho_\theta^{1/2}},
\end{equation}
where $\dket{A} = \sum_{ij} \braket{i|A|j} \ket{i}\ket{j}$. Note that 
\begin{equation}
\dket{v}= \lambda \dket{w},\;\lambda \in \mathbb R.\;
\Longleftrightarrow \;\dket{v}\dbra{w} - \dket{w}\dbra{v} = 0.
\end{equation}
It means that \eqref{eq:ns-5} is equivalent to 
\begin{equation}
(E^{1/2}_{  x } \otimes I) \left(\dket{\rho_\theta^{1/2}}\dbra{L_\theta\rho_\theta^{1/2}} - \dket{L_\theta\rho_\theta^{1/2}}\dbra{\rho_\theta^{1/2}}\right) (E^{1/2}_{  x } \otimes I) = 0.
\end{equation}
Assuming $\rho_\theta = \sum_{k} p_{\theta,k} \ket{\psi_{\theta,k}}\bra{\psi_{\theta,k}}$ ($p_{\theta,k} > 0$), 
\eqref{eq:ns-3} is simplified to 
\begin{equation}
\label{eq:ns-6}
E_x^{1/2} M_{ij} E_x^{1/2} = 0 ,\quad\forall i,j,{  x }.
\end{equation}
where $M_{ij} = \ket{\psi_{\theta,i}}\bra{\psi_{\theta,j}} L_\theta - L_\theta \ket{\psi_{\theta,i}}\bra{\psi_{\theta,j}}$. 
In particular, for rank-one projective measurements $\{E_{  x } = \ket{E_{  x }}\bra{E_{  x }}\}$, \eqref{eq:ns-3} becomes
\begin{equation}
\bra{E_x} M_{ij} \ket{E_x} = 0 ,\quad\forall i,j,{  x }.
\end{equation}

When $\rho_\theta = \ket{\psi_\theta}\bra{\psi_\theta}$ is pure and $p_0 = 1$, the necessary and sufficient conditon becomes
\begin{equation}
E_x^{1/2} M_{00} E_x^{1/2} = 0 ,\quad\forall {  x },
\end{equation}
where $M_{00} = \ket{\psi_\theta}\bra{\psi_\theta} L_\theta - L_\theta \ket{\psi_\theta}\bra{\psi_\theta}$.

When $\rho_\theta = p_\theta \ket{\psi_0}\bra{\psi_0} + (1 - p_\theta) \ket{\psi_1}\bra{\psi_1}$ where $\ket{\psi_{0,1}}$ is independent of $\theta$, the necessary and sufficient condition becomes
\begin{equation}
\label{eq:mixed}
E_x^{1/2} \ket{\psi_0}\bra{\psi_1} E_x^{1/2} = 0 ,\quad\forall {  x },
\end{equation}
because $M_{00} = M_{11} = 0$ and $M_{01} = - M_{10}^\dagger = -\frac{\partial_\theta p_\theta}{p_\theta(1-p_\theta)}\ket{\psi_0}\bra{\psi_1}$.

\section{\label{app:LOCC}QCRB-saturating LOCC}

We first prove \lemmaref{thm:zero-diag} which will become quite useful in constructing QCRB-saturating LOCC: 

\begin{proof}
We only need to prove any two traceless Hermitian matrices $M_1$ and $M_2$ can be simultaneously zero-diagonalized. We first consider the case where $d=2$, i.e. $M_1$ and $M_2$ are $2$-by-$2$ traceless Hermitian matrices. Let 
\begin{equation}
M_{k} = \begin{pmatrix}
a_k & b_k e^{i\phi_k} \\
b_k e^{-i\phi_k} & -a_k\\
\end{pmatrix},
\end{equation}
$k=1,2$, and 
\begin{equation}
U = \begin{pmatrix}
\cos\beta &  -\sin\beta e^{i\alpha}\\
\sin\beta e^{-i\alpha} & \cos\beta \\
\end{pmatrix}.
\end{equation}
Then $U^\dagger M_k U$ has zero diagonal elements is equivalent to
\begin{equation}
a_k (\cos^2\beta - \sin^2\beta) 
= - 2 b_k \cos\beta\sin\beta \cos(\alpha-\phi_k)
\end{equation}
\begin{equation}
\;\Longleftrightarrow\;
\cot 2\beta = - \frac{b_k}{a_k} \cos(\alpha-\phi_k),\quad k=1,2.
\end{equation}
It can be solved by first finding $\alpha$ satisfying
$
b_1a_2 \cos(\alpha-\phi_1) = b_2a_1 \cos(\alpha-\phi_2)
$
and then solving $\beta$ using the equation above. For higher dimension, \lemmaref{thm:zero-diag} can be proven by induction. Suppose \lemmaref{thm:zero-diag} holds for $d \leq \bar d$. Then when $d = \bar d+1$, we only need to find some $\ket{v}$ such that $\braket{v|M_1|v} = \braket{v|M_2|v} = 0$. The rest follows by the induction assumption by simultaneouly diagonalizing $M_1$ and $M_2$ in the $\bar d$ dimensional orthogonal subspace perpendicular to $\ket{v}$. Now we prove the existence of $\ket{v}$. Without loss of generality (WLOG), we assume $M_1\neq 0$ is diagonal, 
\begin{equation}
M_1 = \begin{pmatrix}
\Lambda_1 & 0 \\
0 & \Lambda_2 \\
\end{pmatrix},
\end{equation}
where we divide the Hilbert space into the direct sum of two subspaces and put $M_1$ in a block-diagonal form such that $\Lambda_1 \succ 0$ and $\Lambda_2 \prec 0$. Meanwhile, 
\begin{equation}
M_2 = \begin{pmatrix}
\Sigma_1 & B \\
B^\dagger & \Sigma_2 \\
\end{pmatrix}.
\end{equation}
We can always rescale $M_2$ such that one of the following situations occurs:
\begin{enumerate}[(a),wide,labelwidth=!,labelindent=3pt]
  \item
  $\trace(\Lambda_1)=\trace(\Sigma_1) > 0$ and $\trace(\Lambda_2)=\trace(\Sigma_2) < 0$. Then by the induction assumption, there are $\ket{v_1}$ and $\ket{v_2}$, s.t.
  \begin{align}
  \braket{v_1|\Lambda_1|v_1} &= \braket{v_1|\Sigma_1|v_1} > 0,\\
  \braket{v_2|\Lambda_2|v_2} &= \braket{v_2|\Sigma_2|v_2} < 0.
  \end{align}
  Let $\ket{v} = \cos\beta \ket{v_1 \oplus 0}+ \sin\beta e^{-i\alpha} \ket{0 \oplus v_2}$, we have
  \begin{equation}
  \braket{v|M_1|v} = \cos^2\beta \braket{v_1|\Lambda_1|v_1} 
  + \sin^2\beta \braket{v_2|\Lambda_2|v_2},
  \end{equation}
  \begin{equation}
  \braket{v|M_2|v} = \cos^2\beta \braket{v_1|\Sigma_1|v_1} 
  + \sin^2\beta \braket{v_2|\Sigma_2|v_2} + 2\cos\beta\sin\beta \mathrm{Re}[e^{-i\alpha} \bra{v_1} B \ket{v_2}].
  \end{equation}
  Clearly, there is a solution $(\alpha,\beta)$, s.t. $\braket{v|M_1|v} = \braket{v|M_2|v} = 0$.
  \item $\trace(\Sigma_1) = \trace(\Sigma_2) = 0$. Then by the induction assumption, there are $\ket{v_1}$ and $\ket{v_2}$, s.t.
  \begin{align}
  \braket{v_1|\Lambda_1|v_1} >0&,~~ \braket{v_1|\Sigma_1|v_1} = 0,\\
  \braket{v_2|\Lambda_2|v_2} <0&,~~ \braket{v_2|\Sigma_2|v_2} = 0.
  \end{align}
  Let $\ket{v} = \cos\beta \ket{v_1 \oplus 0}+ \sin\beta e^{-i\alpha} \ket{0 \oplus v_2}$, we have
  \begin{gather}
  \braket{v|M_1|v} = \cos^2\beta \braket{v_1|\Lambda_1|v_1} 
  + \sin^2\beta \braket{v_2|\Lambda_2|v_2},\\
  \braket{v|M_2|v} = 2\cos\beta\sin\beta \mathrm{Re}[e^{-i\alpha} \bra{v_1} B \ket{v_2}].
  \end{gather}
  Clearly, there is a solution $(\alpha,\beta)$, s.t. $\braket{v|M_1|v} = \braket{v|M_2|v} = 0$.
\end{enumerate}
\lemmaref{thm:zero-diag} is then proved. 
\end{proof}

To find the QCRB-saturating LOCC in \thmref{thm:locc}, we only need to find an orthonormal basis which has the structure $E_{x_1,\ldots,x_n} = E_{x_1}^{{\mathfrak s}_1} \otimes E_{x_1,x_2}^{{\mathfrak s}_2} \otimes \cdots \otimes E_{x_1,\ldots,x_n}^{{\mathfrak s}_n}$ and satisfy $\braket{E_{x_1,\ldots,x_n}|\tM|E_{x_1,\ldots,x_n}} = 0$ as well. It can be constructed by the following procedure:
\begin{enumerate}[(1),wide,labelwidth=!,labelindent=3pt]
\item Find an orthonormal basis $\{\ket{E^{{\mathfrak s}_1}_{x_1}}\}_{x_1=1}^{\dim {\mathfrak s}_1}$ which zero-diagonalizes $\tM^{{\mathfrak s}_1} = \trace_{{\mathfrak s}_2\cdots {\mathfrak s}_n}(\tM)$, i.e. $\braket{E^{{\mathfrak s}_1}_{x_1}|\tM^{{\mathfrak s}_1}|E^{{\mathfrak s}_1}_{x_1}} = 0$ for all $x_1$.  
\item Find an orthonormal basis $\{\ket{E^{{\mathfrak s}_2}_{x_1,x_2}}\}_{x_2=1}^{\dim {\mathfrak s}_2}$ which zero-diagonalizes $\tM^{{\mathfrak s}_2}_{x_1} = \trace_{{\mathfrak s}_3\cdots {\mathfrak s}_n}\braket{E^{{\mathfrak s}_1}_{x_1}|\tM|E^{{\mathfrak s}_1}_{x_1}}$.
\item Find an orthonormal basis $\{\ket{E^{{\mathfrak s}_k}_{x_1,\ldots,x_k}}\}_{x_k=1}^{\dim {\mathfrak s}_k}$ which zero-diagonalizes 
\begin{equation}
\tM^{{\mathfrak s}_k}_{x_1,\ldots,x_{k-1}} = \trace_{{\mathfrak s}_{k+1}\cdots {\mathfrak s}_n} \bra{E^{{\mathfrak s}_{k-1}}_{x_1,\ldots,x_{k-1}}}\cdots\bra{E^{{\mathfrak s}_1}_{x_1}} \tM \ket{E^{{\mathfrak s}_1}_{x_1}}\cdots\ket{E^{{\mathfrak s}_{k-1}}_{x_1,\ldots,x_{k-1}}}
\end{equation}
till $k=n$. 
\end{enumerate}
Then one can easily verify 
\begin{equation}
E_{x_1,\ldots,x_n} = E_{x_1}^{{\mathfrak s}_1} \otimes E_{x_1,x_2}^{{\mathfrak s}_2} \otimes \cdots \otimes E_{x_1,\ldots,x_n}^{{\mathfrak s}_n}
\end{equation}
is QCRB-saturating, where 
\begin{equation}
E_{x_1,\ldots,x_k}^{{\mathfrak s}_k} = \ket{E_{x_1,\ldots,x_k}^{{\mathfrak s}_k}}\bra{E_{x_1,\ldots,x_k}^{{\mathfrak s}_k}}.
\end{equation}
Note that the proof of \lemmaref{thm:zero-diag} is constructive. It means the QCRB-saturating LOCC can be calculated directly from matrix $\tM$.

\section{\label{app:dist}Distinguishing two orthogonal quantum states}

In Ref.~\cite{walgate2000local}, the distinguishability of two multipartite orthogonal states $\{\ket{\psi_{0,1}}\}$ via LOCC is shown by writing them as 
\begin{align}
\ket{\psi_0} &= \sum_{(x_1,\ldots,x_n) \in {\mathfrak s}_0} \alpha_{x_1,\ldots,x_n}\ket{x_1}\cdots \ket{x_n}_{x_1,\ldots,x_{n-1}},\\
\ket{\psi_1} &= \sum_{(x_1,\ldots,x_n) \in {\mathfrak s}_1} \alpha_{x_1,\ldots,x_n} \ket{x_1}\cdots \ket{x_n}_{x_1,\ldots,x_{n-1}},
\end{align}
where $x_k \in [\dim({\mathfrak s}_k)]$, ${\mathfrak s}_0 \cap {\mathfrak s}_1 = \emptyset$, $\alpha_{x_1,\ldots,x_n}$ are probability amplitudes and $\braket{x_k|x'_k}_{x_1,\ldots,x_{k-1}} = \delta_{x_k,x'_k}$.
As we can see, it is equivalent to the QCRB-saturating condition for rank-two mixed states with fixed eigenbasis $\ket{\psi_{0,1}}$,
\begin{equation}
\braket{E_x|\psi_0}\braket{\psi_1|E_x} = 0,\text{ for all }\ket{E_x}.
\end{equation}
The LOCC measurement basis $\ket{E_x}$ corresponds to $\ket{x_1}\ket{x_2}_{x_1}\cdots \ket{x_n}_{x_1,\ldots,x_{n-1}}$.

\section{An example of the LOCC protocol}
\label{app:example}

Here we demonstrate our LOCC protocol by considering an open boundary Hamiltonian in a four-qubit system
\begin{equation}
\label{eq:open}
H = \theta \sum_{i=1}^3 \sigma^{x,{\mathfrak s}_i}\sigma^{x,{\mathfrak s}_{i+1}},
\end{equation} 
where $\sigma^{x,i}$ is the Pauli-X matrix acting on the $i$-th qubit and a Dicke state input
\begin{equation}
\label{eq:dicke}
\ket{\psi_{\rm in}
}= \frac{1}{2}(\ket{1000}+\ket{0100}+\ket{0010}+\ket{0001}).
\end{equation} 
The parameter we want to estimate is the Hamiltonian strength $\theta$ which is encoded in $\ket{\psi_\theta}$ through $\ket{\psi_\theta} = e^{-iH}\ket{\psi_{\rm in}
}$ after unit time evolution. 
Numerical search suggests that there is no QCRB-saturating LM and LOCC is necessary. 
As shown in \figref{fig:locc}, we demonstrate our LOCC protocol by directly calculating a set of QCRB-saturating LOCCs for $\theta \in [0,\frac{\pi}{4}]$. 
The tree structure illustrates the choice of measurement basis for qubit ${\mathfrak s}_k$ dependent on the results from $\{{\mathfrak s}_1,\ldots, {\mathfrak s}_{k-1}\}$ via classical communitation. Note that the LOCC protocol illustrated in \figref{fig:locc} is not unique and there could be other LOCCs that are also QCRB-saturating due to remaining degrees of freedom.

\begin{figure*}[htbp]
\includegraphics[width=18cm]{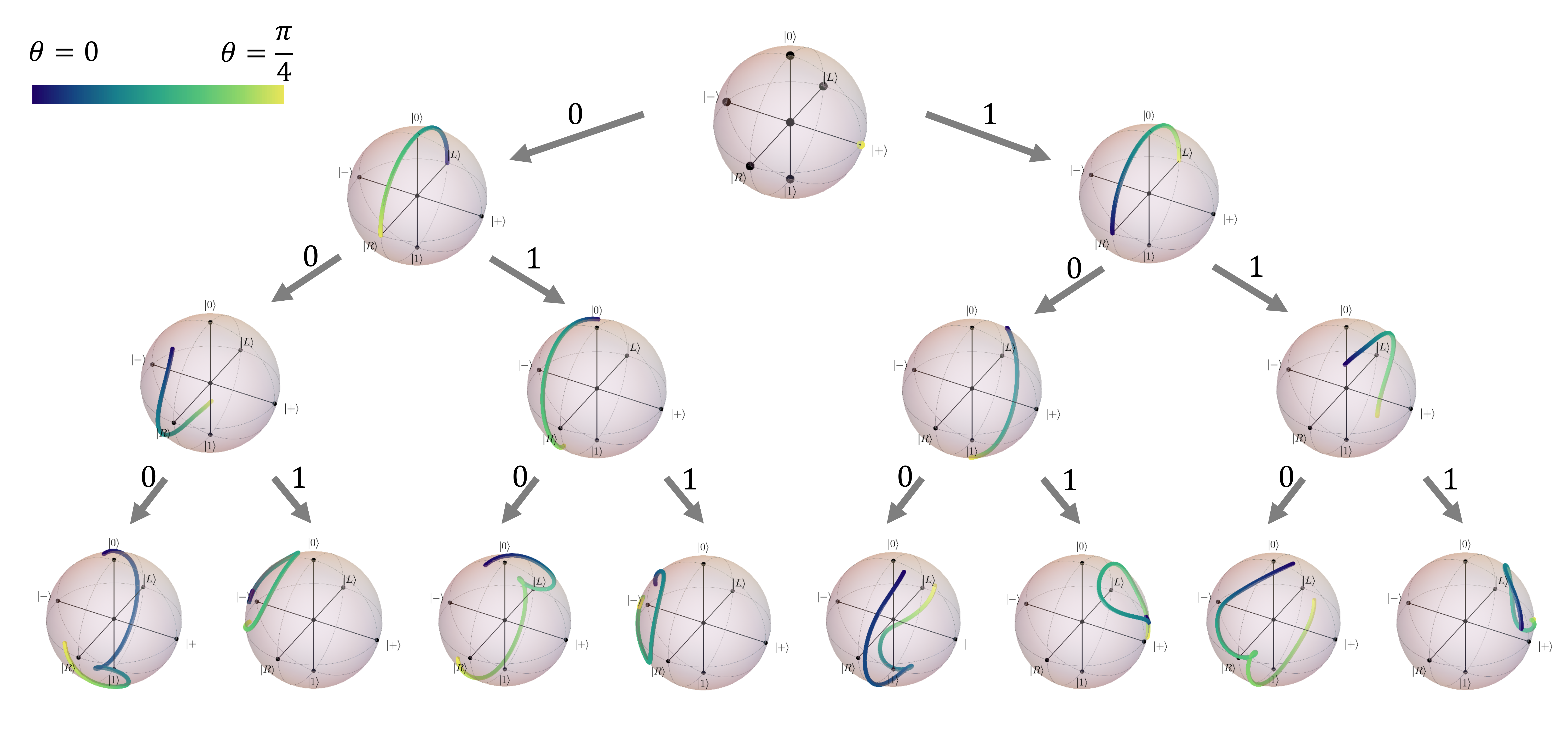}
\caption{\label{fig:locc} Plotting the LOCC measurement basis for each qubit on a Bloch sphere in a four-qubit system described by \eqref{eq:open} and \eqref{eq:dicke}. The eigenstates of Pauli matrices $\ket{0,1}$, $\ket{\pm}$ and $\ket{L,R} = \frac{\ket{0}\pm i\ket{1}}{\sqrt{2}}$ are labeled on the Bloch sphere. Columns from top to bottom each represent measurement on qubit ${\mathfrak s}_1,{\mathfrak s}_2,{\mathfrak s}_3,{\mathfrak s}_4$ and arrows represent how the measurement basis should be chosen based on previous measurement results $x=0,1$. The measurement basis is represented by a point on the surface of the Bloch sphere which corresponds to $x=0$. The color indicates the value of Hamiltonian strength $\theta \in [0,\frac{\pi}{4}]$. Note that the measurement on first qubit does not change with time because $M^{{\mathfrak s}_1} = \trace_{{\mathfrak s}_2 \ldots {\mathfrak s}_4}(M)$ and $\trace_{{\mathfrak s}_2 \ldots {\mathfrak s}_4}(\ket{\psi_\theta}\bra{\psi_\theta}-\frac{1}{16}I)$ are always proportional to $\sigma_z$. }
\end{figure*}

\section{\label{app:no-go-LM} LM is not sufficient to saturate the QCRB 
for multipartite systems
}

Here we want to show LM is not sufficient to saturate the QCRB for sufficiently large $n$ in general sensing scenarios, where $n$ is the number of subsystems. Consider Hilbert space $\mathcal H^{\otimes n}$ where $\dim \mathcal H = d$. For type (i) or (ii) states, let 
\begin{equation}
E_{x_1,\ldots,x_n} = E_{x_1}^{{\mathfrak s}_1} \otimes E_{x_2}^{{\mathfrak s}_2} \otimes \cdots \otimes E_{x_n}^{{\mathfrak s}_n},
\end{equation}
be a QCRB-saturating local measurement satisfying 
\begin{equation}
\label{eq:lm-no-go-1}
(E_{x_1}^{{\mathfrak s}_1} \otimes E_{x_2}^{{\mathfrak s}_2} \otimes \cdots \otimes E_{x_n}^{{\mathfrak s}_n})M(E_{x_1}^{{\mathfrak s}_1} \otimes E_{x_2}^{{\mathfrak s}_2} \otimes \cdots \otimes E_{x_n}^{{\mathfrak s}_n}) = 0,~~~\forall (x_1,\ldots,x_n)\in[r_1]\times \cdots\times [r_n]. 
\end{equation}
where $[r]= \{1,2,\ldots,r\}$. Let $\{\ket{E_{x_j,i_j}^{\fraks_j}}\}_{i_j}$ be a basis of the support of $E_{x_j}^{\fraks_j}$. Then \eqref{eq:lm-no-go-1} implies 
\begin{equation}
\label{eq:lm-no-go-2}
\trace\big(M(F_1^{\fraks_1}\otimes \cdots\otimes F_n^{\fraks_n})\big) = 0,
\end{equation}
for all Hermitian $F_i^{\fraks_i} \in {\rm span}\{\ket{E_{x_i,j_i}^{\fraks_i}}\bra{E_{x_i,j_i}^{\fraks_i}}, \forall x_i, j_i\}$.



We first consider type (i) states, then if we define $\ket{\psi_\theta^\perp} = (1-\ket{\psi_\theta}\bra{\psi_\theta})\ket{\partial_\theta \psi}$,
\begin{equation}
M = \ket{\psi_\theta}\bra{\psi_\theta^\perp} - \ket{\psi_\theta^\perp}\bra{\psi_\theta}.
\end{equation}
Suppose $\mathcal H^{\otimes n}= \mathcal H_1\otimes \mathcal H_2$ where $\mathcal H_1 = \mathcal H^{\otimes m}$ with $m \leq n/2$. Then the reduce matrix $M_r = \trace_{\mathcal H_2}(M)$ after tracing out $\mathcal H_2$ could be an arbitrary traceless anti-Hermitian matrix (up a real factor) by choosing
\begin{gather} 
\ket{\psi_\theta} = \frac{1}{\sqrt{d^{m}}}\sum_{i=1}^{d^{m}} \ket{i}_{\mathcal H_1}\ket{i}_{\mathcal H_2},\\
\ket{\psi_\theta^\perp} = \frac{1}{\sqrt{\trace(-M_r^2)}} \sum_{i,j=1}^{d^{m}} M_{r,ij}\ket{i}_{\mathcal H_1}\ket{j}_{\mathcal H_2}.
\end{gather}
According to \eqref{eq:lm-no-go-2}, $M_r$ has to satisfy
\begin{equation}
\label{eq:lm-no-go-3}
\trace\big(M_r(F_1^{\fraks_1}\otimes \cdots\otimes F_m^{\fraks_n})\big) = 0,
\end{equation}
for all Hermitian $F_i^{\fraks_i} \in {\rm span}\{\ket{E_{x_i,j_i}^{\fraks_i}}\bra{E_{x_i,j_i}^{\fraks_i}}, \forall x_i, j_i\}$. There must exists a (not necessarily orthonormal) basis $\{\ket{e_i^{\mathfrak s_k}}\}_{i=1}^d$ for each $k$ such that 
\begin{equation}
\label{eq:lm-no-go-4}
\bigg( \bigotimes_{k=1}^m \bra{e_{k_i}^{\mathfrak s_k}} \bigg) M_r \bigg( \bigotimes_{k=1}^m \ket{e_{i_k}^{\mathfrak s_k}} \bigg) = 0,\quad \forall i_k = 1,2,\ldots,d. 
\end{equation}

Note that the degree of freedom of an arbitrary traceless anti-Hermitian matrix $M_r$ is $d^{2m} - 1$. But the degree of freedom for matrices satisfying \eqref{eq:lm-no-go-4} is at most $d^{2m} - d^{m} + m d^4$ which is smaller than $d^{2m}$ for large enough $m$. Here $d^m$ is the minimum number of equations of constraints in \eqref{eq:lm-no-go-4} and $m d^4$ is the local freedoms in choosing the basis. Therefore \eqref{eq:lm-no-go-4} could not be satisfied for arbitrary $M_r$, implying that \eqref{eq:lm-no-go-2} could not be satisfied for all possible $M$. For type (ii) states, the same argument holds if we replace $\ket{\psi_\theta}$ and $\ket{\psi_\theta^\perp}$ above with $\ket{\psi_0}$ and $\ket{\psi_1}$.

\section{\texorpdfstring{\eqref{eq:lm-1}}{Eq.~(26)} cannot be distinguished using LM}
\label{app:lm-1}

Consider a two-qubit system. Here we show that 
\begin{equation}
\ket{\psi_\theta} = \frac{1}{\sqrt{2}}(\ket{00}+\ket{1+}),\quad \ket{\psi_\theta^\perp} = \frac{1}{\sqrt{2}}(\ket{01}+\ket{1-})
\end{equation}
cannot be distinguished using LM, i.e. there is no LM $\{E_{x_1}^{\mathfrak s_1}\otimes E_{x_2}^{\mathfrak s_2}\}$ such that 
\begin{equation}
\bra{\psi_\theta}E_{x_1}^{\mathfrak s_1}\otimes E_{x_2}^{\mathfrak s_2}\ket{\psi_\theta}\cdot\bra{\psi_\theta^\perp}E_{x_1}^{\mathfrak s_1}\otimes E_{x_2}^{\mathfrak s_2}\ket{\psi_\theta^\perp} = 0,\quad \forall x_1,x_2. 
\end{equation}
Choose $\ket{\phi_1}$ and $\ket{\phi_2}$ in the support of some $E_{x_1}^{\mathfrak s_1}$ and $E_{x_2}^{\mathfrak s_2}$ respectively. 
Then 
\begin{equation}
\braket{\psi_\theta|\phi_1\otimes\phi_2}\braket{\phi_1\otimes\phi_2|\psi_\theta^\perp} = 0.
\end{equation}
Since the identity operator $I$ is in the span of all $\ket{\phi_2}\bra{\phi_2}$, we must have 
\begin{equation}
\braket{\psi_\theta|\phi_1}\braket{\phi_1|\psi_\theta^\perp} = 0.
\end{equation}
Note that both $\braket{\phi_1|\psi_\theta}$ and $\braket{\phi_1|\psi_\theta^\perp}$ are not zero because $\ket{\psi_\theta}$ and $\ket{\psi_\theta^\perp}$ are entangled. Let $\ket{e^{(2)}_1} \propto \braket{\phi_1|\psi_\theta} $ and $\ket{e^{(2)}_2} \propto \braket{\phi_1|\psi_\theta^\perp}$ be unit vectors. We must have either $\ket{\phi_2} = \ket{e^{(2)}_1}$ or $\ket{\phi_2} = \ket{e^{(2)}_2}$ because $\braket{e^{(2)}_1|\phi_2}\braket{\phi_2|e^{(2)}_2} = 0$. 
Therefore, there are two orthonormal basis $\ket{e^{(1,2)}_{1,2}}$ such that 
\begin{equation}
\braket{e_i^{(1)}\otimes e_j^{(2)}|\psi_\theta}\braket{\psi_\theta^\perp| e_i^{(1)}\otimes e_j^{(2)}} = 0,\quad\forall i,j=1,2. 
\end{equation}
Since $\ket{\psi_\theta}$ and $\ket{\psi_\theta^\perp}$ are both entangled, we must have 
\begin{equation}
\ket{\psi_\theta} = \lambda_1 \ket{e_1^{(1)}}\ket{e_1^{(2)}} + \lambda_2 \ket{e_2^{(1)}}\ket{e_2^{(2)}}, 
\;
\ket{\psi_\theta} = \mu_1 \ket{e_1^{(1)}}\ket{e_2^{(2)}} + \mu_2 \ket{e_2^{(1)}}\ket{e_1^{(2)}}, 
\end{equation}
for some non-zero $\lambda_{1,2}$ and $\mu_{1,2}$. Clearly, \eqref{eq:lm-1} does not have this form. Therefore, they cannot be distinguished using LM. 

On the other hand, there exists an LM such that the QCRB is saturated. For example, 
\begin{gather}
\ket{E_1^{\mathfrak s_1}} = \ket{0},\quad 
\ket{E_2^{\mathfrak s_1}} = \ket{1},\\ 
\ket{E_1^{\mathfrak s_2}} = \cos\frac{\pi}{8}\ket{0} + \sin\frac{\pi}{8}\ket{1},\quad 
\ket{E_2^{\mathfrak s_2}} = -\sin\frac{\pi}{8}\ket{0} + \cos\frac{\pi}{8}\ket{1}.
\end{gather}


\section{There is no QCRB-saturating projective LM for \texorpdfstring{\eqref{eq:3-3}}{Eq.~(27)}}
\label{app:lm-2}

Here we prove there is no QCRB-saturating projective LM when 
\begin{equation}
A = \begin{pmatrix}
{\sqrt{2}}/{2} & 0 & 0 \\
0 & {1}/{2} & 0\\
0 & 0 & 1/2\\
\end{pmatrix},\;
B = \begin{pmatrix}
{\sqrt{2}i}/{2} & 0 & 0 \\
0 & {-i}/{2} & 0\\
0 & 0 & {-i}/{2}\\
\end{pmatrix} 
\end{equation} 
but there is a QCRB-saturating general LM. From the proof of \lemmaref{thm:bi}, we can see that it is equivalent to the statement that there are no unitaries $U$ and $V$ such that \eqref{eq:bipartite-1} and \eqref{eq:bipartite-2} are both satisfied, but there are isometries $U$ and $V$ such that these conditions are satisfied. 

Our goal is to prove there are no unitaries $U$ and $V$ such that \eqref{eq:bipartite-1} is satisfied. Now suppose $U$ and $V$ are unitary. According to the definitions of $C$ and $D$, we have 
\begin{gather}
C_{ij} = (U^\dagger AV)_{ij} = \frac{1}{2}\big( \sqrt{2} \overline{U}_{1i} V_{1j} + \overline{U}_{2i} V_{2j} + \overline{U}_{3i} V_{3j} \big),\\
D_{ij} = (U^\dagger BV)_{ij} = \frac{1}{2}\big( \sqrt{2}i \overline{U}_{1i} V_{1j} - i \overline{U}_{2i} V_{2j} - i \overline{U}_{3i} V_{3j} \big),
\end{gather}
and
\begin{equation}
\begin{split}
C_{ij} \overline{D}_{ij} - D_{ij} \overline{C}_{ij} 
&= \frac{1}{4}(\sqrt{2} \overline{U}_{1i} V_{1j} + \overline{U}_{2i} V_{2j} + U_{i3}v_{3j})(-\sqrt{2}i U_{1i}\overline{V}_{1j} + i U_{2i}\overline{V}_{2j} + i U_{3i}\overline{V}_{3j})
\\& \qquad - (\sqrt{2} U_{1i}\overline{V}_{1j} + U_{2i}\overline{V}_{2j} + U_{3i}\overline{V}_{3j})(\sqrt{2}i \overline{U}_{1i}V_{1j} - i \overline{U}_{2i}V_{2j} - i \overline{U}_{3i}V_{3j})\\
&= - i \abs{\overline{U}_{1i}V_{1j}}^2 + \frac{i}{2} \abs{\overline{U}_{2i}V_{2j} + \overline{U}_{3i}V_{3j}}^2 = 0,\quad \forall i,j.
\end{split}
\end{equation}
First note that the following two transformations do not change \eqref{eq:bipartite-1}:
\begin{enumerate}[(1),wide,labelwidth=!,labelindent=0pt]
\item 
$
U \rightarrow 
\begin{pmatrix}
1 & 0 \\
0 & R \\
\end{pmatrix}
U,\;
V \rightarrow 
\begin{pmatrix}
1 & 0 \\
0 & R \\
\end{pmatrix}
V, 
$
where $R$ is any unitary matrix in $\bC^{2\times 2}$. 
\item 
$
U \rightarrow U D_1,\;
V \rightarrow V D_2, 
$
where $D_1$ and $D_2$ are arbitrary diagonal unitary matrices. 
\end{enumerate}
Therefore, WLOG, we assume $(U_{11},U_{21},U_{31}) = (\mathbb{R}_+,\mathbb{R}_+,0)$.
\begin{gather}
2 \abs{\overline{U}_{11} V_{1j}}^2 = \abs{\overline{U}_{21} V_{2j}}^2 ,\quad \forall j,\\
\Rightarrow~~ 2 \abs{U_{11}}^2 \sum_{j=1}^3 \abs{V_{1j}}^2
= \abs{U_{21}}^2 \sum_{j=1}^3 \abs{V_{2j}}^2,\\
\Rightarrow~~ U_{11} = 1/\sqrt{3},\; U_{21} = \sqrt{2}/\sqrt{3}, \text{~~and~~}
\abs{V_{1j}}^2 = \abs{V_{2j}}^2 =: s_j,\quad \forall j.
\end{gather}
WLOG, assume $V_{2j} = \sqrt{s_j}$, $\forall j$. 

Consider the following two situations:

\noindent Situation (1): $s_j = 0$ for some $j$. Then $\abs{\overline{U}_{3i} V_{3j}}^2 = 0$ and $U_{3i} = 0$ for all $i$, which is not possible.

\noindent Situation (2): $s_j > 0$, $\forall j$. Then for $i\neq 1$, 
\begin{gather}
2  \abs{\overline{U}_{1i} V_{1j}}^2 = \abs{\overline{U}_{2i} V_{2j} + \overline{U}_{3i} V_{3j}}^2 ,\forall j,\\
\Rightarrow~~ 2 \abs{U_{1i}}^2 = \abs{U_{2i}}^2 + \abs{U_{3i}}^2 \frac{1}{s_j}\abs{V_{3j}}^2 +  \frac{2}{\sqrt{s_j}} \cdot \Re[U_{2i} \overline{U}_{3i} V_{3j}], \label{eq:contradict}\\
\Rightarrow~~ 2 \sum_{i>1}\abs{U_{1i}}^2 = \sum_{i>1}\abs{U_{2i}}^2 + \sum_{i>1}\abs{U_{3i}}^2 \frac{1}{s_j}\abs{V_{3j}}^2,\\
\Rightarrow~~ 
\abs{V_{1j}}^2 = \abs{V_{2j}}^2 = \abs{V_{3j}}^2 = 1/3. 
\end{gather}
According to \eqref{eq:contradict}, for some $i\neq 1$ such that $U_{2i} \overline{U}_{3i}$,
\begin{equation}
\Re[U_{2i} \overline{U}_{3i} V_{31}] = \Re[U_{2i} \overline{U}_{3i} V_{32}] =
\Re[U_{2i} \overline{U}_{3i} V_{33}]. 
\end{equation}
We must have $V_{3i} \in \{e^{i\theta_1}/\sqrt{3},e^{i\theta_2}/\sqrt{3}\}$ for all $i$. This is not possible because in that case
\begin{equation}
\sum_{i=1}^3 \overline{V}_{3i} V_{2i} \neq 0,
\end{equation}
contradicting with the requirement that $V$ is unitary. 

On the other hand, let 
\begin{equation}
U = 
\begin{pmatrix}
1/\sqrt{3} & 1/\sqrt{3} & 1/\sqrt{3}  \\
\sqrt{2}/\sqrt{3} & -1/\sqrt{6} & -1/\sqrt{6} \\
0 & -i/\sqrt{2} & i/\sqrt{2} \\
\end{pmatrix},
\quad 
V = 
\begin{pmatrix}
e^{i\pi/4}/2 & e^{i3\pi/4}/2 & -e^{i\pi/4}/2 & -e^{i3\pi/4}/2\\
1/2 & 1/2 & 1/2 & 1/2 \\
1/2 & -1/2 & 1/2 & -1/2\\
\end{pmatrix},
\end{equation}
one can verify in this case both \eqref{eq:bipartite-1} and \eqref{eq:bipartite-2} are satisfied.

\end{document}